\documentclass[,draftclsnofoot,12pt,english,onecolumn]{IEEEtran}
\usepackage{cite}
\usepackage[T1]{fontenc}
\usepackage[latin9]{inputenc}
\usepackage{amsmath}
\usepackage{amssymb}

\makeatletter
\usepackage{amsthm}\usepackage{dsfont}\usepackage{array}\usepackage{mathrsfs}\usepackage{graphicx}
\makeatother

\usepackage{babel}
\def\snr {\mbox{\scriptsize\sf SIR}}
\def\snr {\mbox{\scriptsize\sf SNR}}

\begin{document}
\bibliographystyle{IEEEtran}

\title{Broadcast Channels with Delayed Finite-Rate Feedback: Predict or Observe?}

\author{Jiaming Xu, Jeffrey G. Andrews, Syed A. Jafar    \thanks {J. Xu and J. G. Andrews are with The University of Texas at Austin. S. A. Jafar is with The University of California, Irvine. Contact author is J. G. Andrews: jandrews@ece.utexas.edu. Last modified: \today.}}

\maketitle
\newtheorem{lemma}{Lemma}\newtheorem{theorem}{Theorem}\newtheorem{corollary}{Corollary}\newtheorem{remark}{Remark}

\begin{abstract}
Most multiuser precoding techniques require accurate transmitter channel state information (CSIT) to maintain orthogonality between the users. Such techniques have proven quite fragile in time-varying channels because the CSIT is inherently imperfect due to estimation and feedback delay, as well quantization noise. An alternative approach recently proposed by Maddah-Ali and Tse (MAT) allows for significant multiplexing gain in the multi-input single-output (MISO) broadcast channel (BC) even with transmit CSIT that is completely stale, i.e. uncorrelated with the current channel state.  With $K$  users, their scheme claims to lose only a $\log(K)$ factor relative to the full $K$ degrees of freedom (DoF) attainable in the MISO BC with perfect CSIT for large $K$.  However, their result does not consider the cost of the feedback, which is potentially very large in high mobility (short channel coherence time).  In this paper, we more closely examine the MAT scheme and compare its DoF gain to single user transmission (which always achieves 1 DoF) and partial CSIT linear precoding (which achieves up to $K$).  In particular, assuming the channel coherence time is $N$ symbol periods and the feedback delay is $N_{\rm fd}$ we show that when $N < (1+o(1)) K\log K$ (short coherence time), single user transmission performs best, whereas for $N> (1+o(1)) (N_{\rm fd}+ K / \log K)(1-\log^{-1}K)^{-1} $ (long coherence time), zero-forcing precoding outperforms the other two. The MAT scheme is optimal for intermediate coherence times, which for practical parameter choices is indeed quite a large and significant range, even accounting for the feedback cost.
\end{abstract}

\section{Introduction}

Interference is a key bottleneck in wireless networks and sophisticated interference reduction techniques such as multiuser MIMO \cite{Shamai03, Shamai06}, interference alignment (IA) \cite{Jafar08} and network MIMO \cite{Yu10} are of great interest to researchers and industry.  While these techniques in theory offer substantial multiplexing gains, meaning they support more effectively interference-free streams, they typically require highly accurate transmitter channel state information (CSIT) to achieve said gains. For example, in the case of the broadcast channel with $M$ transmit antennas and $K \geq M$ users, the optimal multiplexing gain with perfect CSIT and CSIR is $M$, which can be achieved by either dirty paper coding (DPC) \cite{Vishwanath03,PViswanath03,Shamai06} or even by suboptimal linear precoders including zero-forcing (ZF) \cite{Jindal07}. Without any CSIT, however, the multiplexing gain collapses to $1$ for i.i.d. channels \cite{Shamai03} or the more general isotropic channel \cite{Jafar05}.

Real systems necessarily have imperfect CSI, particularly at the transmitter.  Because the channel state must be estimated (usually from noisy training symbols), then quantized to a finite-rate value, and finally fedback over a noisy channel in a specified periodic time-slot, the transmitter will have a delayed and noisy estimate of the actual channel state\footnote{Time division duplex (TDD) systems may in principle be able to exploit reciprocity to reduce delay and quantization loss, but such losses still exist.  Furthermore, the reciprocal approach is usually not viable in current systems since the downlink and uplink are nearly independent: the existence of a downlink packet for a user does not imply that the user is also transmitting in the uplink, especially in the same frequency coherence band.}.  In essence, multiuser precoders are relying on a channel prediction in their attempt to separate the users into interference-free channels. In this paper, we only consider the loss due to the delay and finite-rate quantization.

\subsection{Background and Motivation}

The imperfect CSIT issue is by now quite well understood by both academia and industry. Simple multiuser MIMO techniques such as ZF precoding with limited CSI feedback \cite{Jindal06,Yoo07,Ding07,Huang07,Jindal08} have been studied extensively and also implemented in 4G cellular systems \cite{WiMAXBook,LTEBook}. One theoretical observation is that the feedback rate must scale linearly with $\log_2 \snr$ to maintain the full multiplexing gain with partial CSIT \cite{Jindal06}. However, even with this feedback rate, the feedback delay will cause serious degradation when the feedback delay approaches (or exceeds) the channel coherence time, causing multiuser precoding techniques to achieve a lower rate than single user ones due to multiplexing gain loss, regardless of the feedback rate \cite{ZhaKou11}.  Primarily because of this sensitivity, multiuser MIMO techniques (also called SDMA) have been largely disappointing in the field and it is widely agreed they are only of use at very low mobility (pedestrian speeds at most) \cite{Ghosh10,Ghosh11}.


A clever recent work \cite{Tse10} gets around this apparently fundamental delay limitation by instead exploiting \emph{previous} channel observations to increase the multiplexing gain through a novel feedback, transmission, and receiver cancellation scheme.  This technique, which we term the {\it Maddah-Ali-Tse (MAT) scheme}, achieves a multiplexing gain in a MISO broadcast channel of $\frac{K}{1+ \frac{1}{2} + \ldots + \frac{1}{K}}$ which is $K/\log(K)$ for large $K$ even with completely outdated CSIT.  This is nearly as good as the multiplexing gain of $K$ for perfect CSIT precoding schemes. It was subsequently shown in \cite{Jafar10} that a similar conclusion holds for the the X and interference channels using a novel retrospective interference alignment scheme. The guiding principle behind these results is that with significant feedback delay, channel ``predictions'' are bound to fail, but channel ``observations'' can be exploited via eavesdropping and feedback to retroactively remove interference.

A potentially important consideration for these two approaches is the resources they consume on the feedback channel. Both \cite{Tse10} and  \cite{Jafar10} assumed a cost-free infinite rate feedback channel.  The MAT scheme is primarily of interest for mobile scenarios, since with (very) low mobility the conventional channel prediction techniques that achieve the full $K$ DoF can be used.  Therefore, the MAT scheme still inherently requires frequent and accurate channel state feedback, and its main benefit is that it is robust to the feedback delay. The goal of this paper is to determine how much gain (if any) is possible with an outdated/observed CSIT approach, while correctly accounting for the unavoidable feedback channel overhead.

\subsection{Summary of Main Result}
The main technical contribution of the present paper is to determine the \emph{net DoF} provided by the MAT scheme \cite{Tse10} for a $K$ user MISO broadcast channel.  The net DoF is the prelog capacity term remaining after subtracting off the feedback DoF consumed (which depends on the feedback rate). This can then be compared to two other baseline techniques: (i) the no CSIT single user transmitter that always gets 1 DoF and does not require feedback, and (ii) the partial CSIT zero-forcing precoder that gets up to the full $K$ DoF when the CSIT is sufficiently current and accurate. Although many other techniques could be chosen and compared to MAT, these two provide an instructive baseline comparison, and we conjecture that other single user and multiuser precoding schemes would result in a very similar tradeoff.

We start with the $K=2$ user MAT scheme, proceed to $K=3$ and finally provide a general result for all $K$.  The approach leading to the net DoF result proceeds in two steps in each case.  First, we derive the multiplexing gain (DoF) with finite rate feedback as a function of a parameter $\alpha > 0$, where $\alpha \to 0$ is bounded rate CSIT and $\alpha \to \infty$ is perfect CSIT.  The DoF loss relative to what is reported in \cite{Tse10} is zero as long as the feedback amount is sufficiently large. For example, for $K=2$ one requires $\alpha \geq 1$ to achieve the $\frac{4}{3}$ DoF that the MAT scheme promises.  Second, we determine how the sum feedback overhead increases in units of DoF as a function of $\alpha$ and subtract that from the DoF gain found in step 1.

Interestingly, we find that there are regimes where each of the three techniques is the best one, where the regimes are defined by the coherence time $N$ and feedback delay $N_{\rm fd}$ in units of symbol times. Specifically, when the coherence time $N \le (1+o(1)) K \log K $, there is no value of $\alpha$ that allows a \emph{net} increase in DoF from the MAT scheme, i.e. the net DoF is strictly less than $1$ for any $K$, which can be achieved by single-user transmission. Furthermore, when $N>(1+o(1)) ( N_{\rm fd}+ K / \log K)(1-\log^{-1}K)^{-1}$, there is no value of $\alpha$ that allows the MAT scheme to outperform ZF precoding. However, the MAT scheme does provide a net increase in DoF for the optimal value of $\alpha$ in between these two extremes. The main result and this tradeoff is summarized in Fig.~\ref{FigDoF}.

The intuition behind this result is straightforward. The feedback rate for the MAT scheme must be held low in order to not overwhelm the forward direction rate gain. But for a low feedback rate, the resulting channel quantization error becomes large and the MAT scheme fails to work well.  This primarily applies to the high mobility (short coherence time) scenario, since feedback in that case must be frequent.  For sufficiently long coherence times, the feedback delay problem recedes and eventually the conventional orthogonalizing precoders became viable, which approach and eventually achieve the full $K$ DoF. The MAT scheme fills a useful niche for moderate mobility/coherence times, which appears to be a quite broad and relevant regime for reasonable parameter choices.  For example, using a standard LTE air interface with four transmit antennas and a carrier frequency of $2.1$GHz, we find that the MAT scheme is preferable to ZF and single user transmission for velocities ranging from about 27 km/hr up to airplane-type speeds.


\section{System Model} \label{Model}

A MISO broadcast channel with $M$ transmit antennas and $K$ single antenna receivers is considered. In this paper, we assume $M=K$ for simplicity\footnote{Diversity gain can be achieved when K > M (user selection) or M > K (antenna selection), but those only affect the SNR and we are interested in the multiplexing gain. Furthermore, when $K \gg M$, the opportunistic beamforming will achieve nearly optimal degrees of freedom $M$ with small feedback overhead \cite{Hassibi05}, but $K$ needs to grow with $\snr$ which is not assumed in this paper.}. In a flat fading environment, this channel can be modeled as
\begin{align}
y_r [t] = \mathbf{h}^\ast_r [t] \mathbf{x}[t] + z_r[t], \quad r=1, \ldots, K,
\end{align}
where $y_r[t] $ is the received signal of receiver $r$ at symbol time $t$, $\mathbf{x}[t] \in \mathcal{C}^{M \times 1}$ is the transmit signal with the average power constraint $\mathbb{E} [ \mathbf{x}^\ast [t] \mathbf{x} [t]] \le P$, and $z_r[t] \sim \mathcal{CN} (0,1)$ is the additive white Gaussian noise. The channel state vector of receiver $r$ is denoted by $ \mathbf{h}^\ast_r [t] \in \mathcal{C}^{1 \times M}$ and the channel state matrix is defined as $ \mathbf{H} [t] = [ \mathbf{h}_1 [t], \ldots, \mathbf{h}_K [t]]$. The channel is assumed to be block fading: $\mathbf{H} [t] $ remains constant over a block of $N$ symbols, and is comprised of i.i.d. unit variance complex Gaussian random variables for each block. It follows that the $\mathbf{H} [t] $ is full rank with probability $1$.

We consider a delayed finite-rate feedback model. Each receiver is assumed to have an instantaneous and perfect knowledge of its own channel vector $\mathbf{h}_r [t]$. It then quantizes its channel vector to $Q$ bits and feeds back the bits perfectly to the BS with delay of $N_{\rm fd}$ symbols. Notice that we also assume that the receivers obtain the channel state of all other receivers via broadcasting in the forward channel from BS.

The channel state quantization is performed using a fixed vector quantization codebook that is known to the transmitter and all receivers. The codebook $\mathcal{C}$ consists of $2^Q$ $M$-dimensional unit norm vectors: $\mathcal{C}= \{ \mathbf{w}_1,\ldots, \mathbf{w}_{2^Q} \}$. The receiver quantizes its channel vector to the closest quantization vector, i.e., the quantization index at time $t$ is
\begin{align}
q [t] = \arg \min_{i=1,\ldots, 2^Q} \sin^2 \left( \angle \left( \mathbf{h}_r [t], \mathbf{w}_i \right) \right), \nonumber
\end{align}
Note that only the \emph{direction} of the channel vector is quantized and fed back, and no information regarding the channel magnitude is conveyed to the transmitter. Also, in this paper, we consider the optimal codebook over any vector quantization codebook.


The key performance metrics considered in this paper are the {\it degrees of freedom} (also known as the {\it multiplexing gain}), the {\it feedback overhead} and the {\it net DoF}. Let $R(P)$ denote the total average throughput with transmit power $P$. The multiplexing gain with $K$ receivers is defined as
\begin{align}
\text{DoF} (K):= \lim_{P \to \infty} \frac{R(P)}{ \log_2 P}. \label{EquationDefDoF}
\end{align}
The DoF is the prelog of the capacity, and is the number of equivalent channels that carry rate $\log_2 P$ at high SNR.
Let $F(P)$ denote the total feedback rate, then the feedback overhead with $K$ receivers is formally defined as
\begin{align}
\text{FB} (K):= \lim_{P \to \infty} \frac{ F(P)}{\log_2 P}, \nonumber
\end{align}
which measures how quickly the feedback rate increases with $\log_2 P$.
Finally, we define the net multiplexing gain as
\begin{align}
\widehat{\text {DoF} } (K) := \text{DoF} (K) - \text{FB}(K). \nonumber
\end{align}
The net DoF  makes explicit the feedback cost, which is quite important when comparing approaches that require differing amounts of feedback, as in this paper.

\section{BACKGROUND} \label{Background}
In this section, for clarity, we briefly summarize two key previous results, and introduce the traditional zero-forcing as well as modifying the MAT scheme to use finite rate feedback.

\subsection{Multiplexing Gains with Outdated CSIT but no Quantization Error}
The DoF with outdated CSIT but no quantization error was given by the following Theorem \cite{Tse10}.
\begin{theorem} \cite{Tse10} \label{TheoremDOFNoQuantization}
The optimal multiplexing gain with outdated CSIT is
\begin{align}
\text{DoF}^\star (K) = \frac{K}{1+ \frac{1}{2} + \ldots + \frac{1}{K}}.
\end{align}
\end{theorem}
The optimal DoF is achieved by the MAT scheme which was also introduced in \cite{Tse10}, which works as follows. The information symbols intended for a particular receiver can be overheard by other receivers. Even with outdated CSIT fed back by the receiver, the transmitter can exploit this overheard side information to create future transmissions which are simultaneously useful for more than one receiver. However, the result in \cite{Tse10} implicitly assumes the feedback is free.

\subsection{Optimal Vector Quantization }
We now briefly review some basic results on optimal vector quantization in MISO broadcast channel from \cite{Aazhang03,Zhou05,YeungLove07,Jindal06}. Let $\hat{\mathbf{e}}$ denote the quantization of the channel vector $\mathbf{h}$ and $\theta$ denote the angle between $\mathbf{h}$ and $\mathbf{\hat{e}}$. For the optimal vector codebook, the lower and upper bound to the expected value of quantization error are given by \cite{Jindal06}
\begin{align} \label{LemmaQuantization}
\frac{M-1}{M} 2^{-\frac{Q}{M-1}} < \mathbb{E}_{\mathbf{H}} \left[ \sin^2 \theta \right] < 2^{-\frac{Q}{M-1}}.
\end{align}
Thus, the bounds are tight when the number of transmit antennas $M$ is large.

\subsection{Zero Forcing with Delayed Finite-Rate Feedback}
In the slow fading scenario where the feedback delay $N_{\rm fd}$ is smaller than the coherence time $N$, channel prediction schemes can be adopted. In this paper, the zero-forcing scheme (ZF) is considered, since multiplexing gain is of interest and the zero-forcing scheme can achieve the full $K$ multiplexing gain with perfect CSIT \cite{Jindal06,Jindal07}. The ZF scheme proceeds as follows: at the beginning of each block, each user quantizes its own channel vector to $B$ bits and feedbacks the quantization index. The transmitter receives the feedback with delay $N_{fd}$ and uses zero-forcing precoding based on the quantized and delayed CSI over the remaining $N-N_{fd}$ symbol times. According to the analysis in \cite{Jindal06}, if the quantization rate $Q$ is scaled as $Q= \alpha (K-1) \log_2 P$ for $\alpha >0$, the DoF under the ZF scheme is
\begin{align}
\text{DoF}_{\text{ZF}} (K) = \left( 1- \frac{N_{\rm fd} }{N} \right)  (\alpha \land 1) K,
\end{align}
and the feedback overhead is $ \frac{\alpha K(K-1)}{N}$. Therefore, the net DoF is
\begin{align}
\widehat{\text {DoF} }_{\text{ZF}} (K)=K \left( 1- \frac{ (\alpha \land 1) N_{\rm fd} +  (K-1)\alpha  }{N} \right), \label{EqNetDoFZF}
\end{align}
which has a maximum value $K( 1- \frac{N_{\rm fd} +  (K-1)}{N})$ with $\alpha=1$.


\subsection{Exploiting Outdated CSIT via Delayed Finite-Rate Feedback} \label{sectionMAT}
In this subsection, we consider the MAT scheme with finite rate CSI feedback. Since the interference-limited case and multiplexing gain are of interest, Gaussian noise is omitted for simplicity. Also, the feedback delay is assumed to be identical to the coherence time in the description, i.e., $N_{\rm fd}=N$, but the MAT scheme can be extended to the general case. Let $u_r$ and $v_r$ denote the symbols from two independently encoded Gaussian codewords intended for receiver $r$. The transmission scheme consists of two phases, which take three symbol times over three consecutive blocks in total as shown in Fig.~\ref{FigSysModel}.

{\it Phase one: Feeding the Receivers}. This phase has two symbol times. In the first symbol times of block 1, the BS transmits the two symbols, $u_A$ and $v_A$, intended for the receiver $A$, i.e., $\mathbf{X} [1] = [ u_A, v_A]^t$. At the receivers, we have
\begin{align}
y_A[1]&= h_{A1}^\ast [1] u_A + h_{A2}^\ast [1] v_A  := L_A (u_A, v_A) [1] , \nonumber \\
y_B[1]&= h_{B1}^\ast [1] u_A + h_{B2}^\ast [1] v_A  := L_B (u_A, v_A) [1] . \nonumber
\end{align}
At the same time, receiver $B$ measures the channel and obtains perfect knowledge of its channel vector. It then quantizes the channel vector and feeds back the quantization index. Let $\mathbf{\hat{e}}_B [1]=[ \hat{e}_{B1}[1], \hat{e}_{B2}[1]]^t $ denote the quantized channel vector of receiver $B$.

Similarly, in the second symbol times of the block 2, the BS transmits the two symbols, $u_B$ and $v_B$, intended for the receiver $B$, i.e., $\mathbf{X} [2] = [ u_B, v_B]^t$. At the receivers, we have
\begin{align}
y_A[2]&= h_{A1}^\ast [2] u_B + h_{A2}^\ast [2] v_B := L_A(u_B,v_B) [2] , \nonumber  \\
y_B[2]&= h_{B1}^\ast [2] u_B + h_{B2}^\ast [2] v_B := L_B(u_B,v_B) [2] . \nonumber
\end{align}
At the same time, receiver $A$ feeds back its quantized channel vector $\mathbf{\hat{e}}_A [2]=[ \hat{e}_{A1}[2], \hat{e}_{A2}[2] ]^t $.

A key observation is that because the channel matrix $\mathbf{H}[1]$ is full rank with probability $1$, if the receiver $A$ has the equation $ L_B (u_A, v_A) [1]$ overheard by the receiver $B$, then it has enough equations to solve for its own symbols $u_A$ and $v_A$. The same story holds for receiver $B$. Therefore, the goal of the second phase is to swap these two overheard equations through the transmitter.

{\it Phase two: Swapping the overheard Equations}. This phase has one symbol time in the block 3. Since we assume the feedback delay is one block duration, the quantized channel vectors at symbol time $1$ and $2$ are available at the transmitter. Thus, the BS can transmit $\hat{L}_B (u_A, v_A) [1]+ \hat{L}_A(u_B,v_B) [2]$, i.e., $\mathbf{X} [3] = [\hat{L}_B (u_A, v_A) [1]+ \hat{L}_A(u_B,v_B) [2], 0]^t$, where $\hat{L}_B (u_A, v_A) [1]= \hat{e}_{B1}^\ast [1] u_A + \hat{e}_{B2}^\ast [1] v_A$ and $ \hat{L}_A(u_B,v_B) [2]= \hat{e}_{A1}^\ast [2] u_B + \hat{e}_{A2}^\ast [2] v_B$.  The receiver observes
\begin{align}
y_A[3]&= h_{A1}^\ast [3]  \left( \hat{L}_B (u_A, v_A) [1] +  \hat{L}_A(u_B,v_B) [2] \right), \nonumber \\
y_B[3]&= h_{B1}^\ast [3]  \left( \hat{L}_B (u_A, v_A) [1] +   \hat{L}_A(u_B,v_B) [2] \right). \nonumber
\end{align}

The transmission scheme is summarized in Fig.~\ref{FigSysModel}. Putting all these received equations together in matrix form, for receiver A we have
\begin{align}
\begin{bmatrix}
y_A[1] \\
y_A[2] \\
y_A[3]
\end{bmatrix}
 = \begin{bmatrix}
h_{A1}^\ast [1] & h_{A2}^\ast [1] \\
0 & 0 \\
h_{A1}^\ast [3] \hat{e}_{B1}^\star [1] & h_{A1}^\ast [3] \hat{e}_{B2}^\star [1]
\end{bmatrix}
\begin{bmatrix}
u_A \\
v_A \\
\end{bmatrix}
+ \begin{bmatrix}
0 & 0 \\
h_{A1}^\ast [2] & h_{A2}^\ast [2] \\
h_{A1}^\ast [3] \hat{e}_{A1}^\ast [2] & h_{A1}^\ast [3] \hat{e}_{A2}^\ast [2]
\end{bmatrix}
\begin{bmatrix}
u_B \\
v_B \\
\end{bmatrix},
\end{align}
rewritten in a simpler form as
\begin{align} \label{EquationChannelTwoUser}
\mathbf{y}_A =
\begin{bmatrix}
y_A[1] \\
y_A[2]/|\mathbf{h}_A [2]| \\
y_A[3]/ | h_{A1} [3]|
\end{bmatrix}
=
\mathbf{H}^{\ast}_A
\begin{bmatrix}
u_A \\
v_A \\
\end{bmatrix}
+ \mathbf{I}^{\ast}_B
\begin{bmatrix}
u_B \\
v_B \\
\end{bmatrix},
\end{align}
where $\mathbf{I}^{\ast}_B$ denotes the interference from transmitting symbols intended for receiver $B$. This equation has an interference alignment interpretation. The received signal $ \mathbf{y}_A $ lies in the $3$-dimensional vector space. The intended symbols $u_A$ and $v_A$ of receiver $A$ are sent along the vectors $\left[h_{A1}^\ast [1] ,0, h_{A1}^\ast [3] \hat{e}_{B1}^\ast [1] \right]^t$ and $\left[ h_{A2}^\ast [1], 0, h_{A1}^\ast [3] \hat{e}_{B2}^\ast [1]  \right]^t $ respectively. The unintended symbols $u_B$ and $v_B$ are sent along the vectors $ \left[ 0, h_{A1}^\ast [2], h_{A1}^\ast [3] \hat{e}_{A1}^\ast [2]\right]^t$ and $ \left[ 0,   h_{A2}^\ast [2], h_{A1}^\ast [3] \hat{e}_{A2}^\ast [2] \right]^t$ respectively.

If the quantization error is $0$, then the quantized channel vector $\mathbf{\hat{e}}_A [2]$ is in the same direction as the true channel vector $ \mathbf{h}_A [2]$ and thus $u_B$ and $v_B$ are sent along the same direction. Therefore, by zero-forcing the interference, the receiver $A$ has $2$ interference-free dimensions. Also, due to the CSI broadcast from BS, the receiver $A$ knows the quantized channel vector $\mathbf{\hat{e}}_B [1]$ of receiver $B$. Therefore, the receiver A can recover its intended symbols $u_A$ and $v_A$ by solving two independent linear equations. Similarly, for the receiver $B$, it also has $2$ interference-free dimensions and can recover its intended symbols $u_B$ and $v_B$. All in all, we recover $4$ symbols using $3$ symbol times and thus $\frac{4}{3}$ degrees of freedom is achieved.

However, it is easy to see that with finite rate feedback, the quantization error is in general nonzero and thus less than $\frac{4}{3}$ degrees of freedom will be achieved in general.

\section{Main Results} \label{MainResults}
In this section, the impact on the DoF of accounting for finite rate feedback is investigated for the MAT scheme, and then compared to single user and ZF transmission.  We first start with the two user case, then move to three users and finally generalize to the $K$ user case. The proof techniques are essentially the same, but the details of the $K=2$ and $K=3$ cases are instructive to understand the general case.

\subsection{Two User Case, $K=2$}
We first consider the two user case as introduced in Section \ref{sectionMAT}. The following lemma relates the multiplexing gain to the singular values of interference matrix $\mathbf{I}_B^\ast$.

\begin{lemma} \label{LemmaDOFTwoUser}
For the two user case, the DoF with the MAT scheme is
\begin{align}
\text{DoF}_{\text{MAT}} (2) =  \frac{2}{3}- \frac{2}{3} \left( \lim_{P \to \infty}  \mathbb{E} \left[ \frac{ \log_2 ( \sigma_2^2) }{  \log_2 P }  \right] \vee (-1) \right) \label{EquationDoFSimulation},
\end{align}
where $\sigma_2 > 0$ is the second largest singular value of $\mathbf{I}_B^\ast$.
\end{lemma}
\begin{proof}
See Appendix \ref{ProofLemmaDOFTwoUser}.\\
\end{proof}

From Lemma \ref{LemmaDOFTwoUser}, we can derive the following theorem which relates the multiplexing gain to the quantization accuracy $Q$, which is the feedback rate per user per feedback interval.

\begin{theorem} \label{TheoremDOFTwoUser}
For the two user case, if the quantization rate $Q$ is scaled as $Q=\alpha \log_2 P$ for any $\alpha>0$, the multiplexing gain under the MAT scheme is $ \frac{2 (1+ \alpha)}{3}  \land \frac{4}{3}$.
\end{theorem}
\begin{proof}
See Appendix \ref{ProofTheoremDOFTwoUser}.\\
\end{proof}

In Theorem \ref{TheoremDOFTwoUser}, we require $\alpha > 0$ in order to ensure that $Q \to \infty$ as $ P \to \infty$.  However, if $\alpha \to 0$ then the asymptotic multiplexing gain is $\frac{2}{3}$ in the bounded $Q$ case.  Therefore we conclude that for all $\alpha$ the multiplexing gain is between $\frac{2}{3}$ and $\frac{4}{3}$.  When $\alpha \ge 1$, the full multiplexing gain $\frac{4}{3}$ can be achieved, which recovers the result in \cite{Tse10}.  Finally, we note that when $\alpha < \frac{1}{2}$, the multiplexing gain becomes less than $1$ which is less than the no CSIT case (i.e. single user transmission).


The following lemma characterizes the feedback overhead required by the MAT scheme.

\begin{lemma} \label{TheoremFeedbackTwoUser}
For the two user case, if the quantization rate $Q$ is scaled as $Q=\alpha \log_2 P$ for $\alpha >0 $, then the feedback overhead under the MAT scheme is $\frac{2}{3N}\alpha$.
\end{lemma}
\begin{proof}
From the transmission scheme, it can be seen that receiver $A$ has to feedback its channel vector at symbol time $2$ of the block 2 to align the interference; while receiver $B$ has to feedback its channel vector at symbol time $1$ of the block 1 to align the interference. Therefore, over total $3$ blocks, $2$ channel vectors must be fedback. It follows then the aggregate feedback rate is $\frac{2}{3N}\alpha \log_2 P $ bits/slot and thus the feedback overhead is $\frac{2}{3N}\alpha$.
\end{proof}

Combining Theorem \ref{TheoremDOFTwoUser} and Lemma \ref{TheoremFeedbackTwoUser} gives the net DoF as
\begin{eqnarray}
\widehat{\text {DoF} }_{\text{MAT}} (K= 2) &=& \text{DoF} (K) - \text{FB}(K) \nonumber \\
        &=& \frac{2}{3} \left( (1+ \frac{N-1}{N} \alpha) \land (2  - \frac{\alpha}{N}) \right),
\end{eqnarray}
which has a maximum of $\frac{2(2N-1)}{3N}$.

For comparison, let us consider the single user transmission (SISO) and zero-forcing scheme. For single user transmission, it is easy to see that the net DoF is $1$, since no CSI feedback is needed. For zero-forcing, according to (\ref{EqNetDoFZF}), the maximum net DoF with $\alpha=1$ is
\begin{align}
\widehat{\text {DoF} }_{\text{ZF}} (K=2)=2 \left( 1- \frac{  N_{\rm fd} + 1  }{N} \right).
\end{align}

When $N \le 2$, the maximum net DoF with MAT is no greater than $1$ and thus the MAT scheme cannot provide a net gain in DoF compared to SISO for $K=2$ and any value of $\alpha$. When $N > 3N_{\rm fd} +2 $, the maximum net DoF with MAT is less than that with ZF, and thus the MAT scheme cannot provide a net gain for $K=2$ and any value of $\alpha$ either. However, when $2 \le N \le 3N_{\rm fd} +2$, the MAT scheme can provide a net gain in DoF compared to SISO and ZF for $K=2$ with optimal value of $\alpha=1$. Fig.~\ref{FigNetDoF2} present numerical results of the maximum net DoF for the three techniques respectively. The feedback delay is assumed to be $N_{\rm fd}=100$ for clarity. The actual feedback delay depends on the system design. For example, in a narrowband channel with symbol rate $100$ KHZ, $N_{\rm fd}=100$ means the feedback delay is $1$ msec. From the figure, it can be seen that the MAT scheme does provide a net DoF gain when the $ 2 \le N \le 302$, which is a significant range.





\subsection{Three User Case}
In this subsection, we extend to the three user case. According to \cite{Tse10}, the MAT scheme takes $11$ symbol times in $11$ consecutive blocks. The received signal of user $A$ over the total $11$ symbol times can be written in matrix form as
\begin{align}
\mathbf{y}_A = \mathbf{H}_A^\ast \mathbf{u}_A
+ \mathbf{I}_{B,C}^\ast [ \mathbf{u}_B^t, \mathbf{u}_C^t]^t,
\end{align}
where $\mathbf{u}_{A}$ is a $ 6 \times 1 $ vector of symbols intended only  for user $A$,
$\mathbf{H}_A^\ast$ is a $11 \times 6$ matrix with rank $6$, and $\mathbf{I}_{B,C}^\ast$ is a $11 \times 12$ interference matrix. Notice that  $\mathbf{I}_{B,C}^\ast$ must contain $2$ zero row vectors, since among the $11$ symbol times, $2$ symbol times are used for transmitting the symbols only intended for user $A$ and thus there is no interference.

Without quantization error, the rank of $\mathbf{I}_{B,C}^\ast$ is $5$ and thus the interference lies in a $5$-dimensional subspace.  Therefore, $6$ degrees of freedom among $11$ symbol times can be achieved by zero forcing the interference for each user. With quantization error, the interference will spill out of the $5$-dim subspace and the zero-forcing scheme cannot eliminate all the interference. Similar to the two user case, we have the following lemma, which relating the multiplexing gain with the singular values of $\mathbf{I}_{B,C}^\ast$:$\{ \sigma_1, \ldots, \sigma_{11}\}$.

\begin{lemma} \label{LemmaDOFThreeUser}
For the three user case, the multiplexing gain under the MAT scheme is
\begin{align}
\text{DoF}_{\text{MAT}}(3) = \frac{6}{11} - \frac{3}{11} \sum_{i=6}^9 \left(  \lim_{P \to \infty}   \frac{  \mathbb{E} \left[\log_2 (\sigma_i^2)\right] }{  \log_2 P } \vee (-1) \right). \label{EquationDoFSimulationThreeUser}
\end{align}
\end{lemma}
\begin{proof}
See Appendix \ref{ProofLemmaDOFThreeUser}.
\end{proof}

From Lemma \ref{LemmaDOFThreeUser}, we can derive the following theorem which relates the multiplexing gain to the quantization accuracy $Q$.
\begin{theorem} \label{TheoremDOFThreeUser}
For the three user case, if the quantization rate $Q$ is scaled as $Q= 2 \alpha \log_2 P$ for $\alpha>0$, the multiplexing gain under the MAT scheme is $ \frac{6  (1+2 \alpha)}{11} \land \frac{18}{11}$.
\end{theorem}
\begin{proof}
See Appendix \ref{ProofTheoremDOFThreeUser}.
\end{proof}


The feedback overhead is given by the following lemma.
\begin{lemma}\label{TheoremFeedbackThreeUser}
For the three user case, if the quantization rate $Q$ is scaled as $Q= 2 \alpha \log_2 P$ for $\alpha>0$, then the feedback overhead under the MAT scheme is $\frac{30}{11N}\alpha$.
\end{lemma}
\begin{proof}
From the transmission scheme, it can be seen that the key role played by the feedback is to ensure the interference in the future transmissions always lies in the $5$-dimensional subspace. Therefore, each receiver have to feedback its channel vector $5$ times to inform the transmitter its  specific $5$-dimensional subspace. Thus, over total $11N$ symbol times, $15$ channel vectors must be fedback in total. It follows then the total feedback rate is $\frac{30}{11N}\alpha \log_2 P $ bits/slot.
\end{proof}

Combining Theorem \ref{TheoremDOFThreeUser} and Lemma \ref{TheoremFeedbackThreeUser} gives the net DoF as
\begin{eqnarray}
\widehat{\text {DoF} }_{\text{MAT}} (K = 3) &=& \text{DoF} (K) - \text{FB}(K) \nonumber \\
        &=& \frac{6}{11} \left( (1+ \frac{2N-5}{N} \alpha) \land (3-\frac{5}{N}\alpha) \right),
\end{eqnarray}
which has a maximum of $\frac{6(3N-5)}{11N}$.

The maximum net DoF with SISO is still $1$ and according to (\ref{EqNetDoFZF}), the maximum net DoF with ZF for $K=3$ is
\begin{align}
\widehat{\text {DoF} }_{\text{ZF}} (K=3)=3 \left( 1- \frac{  N_{\rm fd} + 2  }{N} \right).
\end{align}

When $N \le \frac{30}{7}$, the maximum net DoF with MAT is no greater than $1$ and thus the MAT scheme cannot provide a net gain in DoF compared to SISO for $K=3$ and any value of $\alpha$. When $N > \frac{11 N_{\rm fd} +12}{5}$, the maximum net DoF with MAT is less than for ZF.  The useful range therefore is
\[
\frac{30}{7} \le N \le \frac{11 N_{\rm fd} +12}{5},
\]
for which the MAT scheme can achieve a net gain. Fig.~\ref{FigNetDoF3} shows the maximum net DoF for the three techniques for a feedback delay of $N_{\rm fd}=100$. The MAT scheme provides a net DoF gain in this case for $5 \le N \le 222$.

\subsection{General $K$ User Case}
In this subsection, we generalize the previous results to the $K$ user case. Let $\frac{K}{1+ \frac{1}{2} + \ldots + \frac{1}{K}}= \frac{K D}{T}$, where $D, T \in \mathcal{N}$. According to \cite{Tse10}, the transmission scheme for the $K$ user case takes $T$ symbol times in the $T$ consecutive blocks, and the received signal for user $A$ over $T$ symbol times can be written in matrix form as
\begin{align} \label{EquationChannelKUser}
\mathbf{y}_A = \mathbf{H}_A^\ast \mathbf{u}_A
+ \mathbf{I}_{/ A}^\ast \mathbf{u}_{/A},
\end{align}
where $\mathbf{u}_{A}$ is a $ D \times 1 $ vector of symbols intended for user $A$, $\mathbf{u}_{/A}$ is a $(K-1)D \times 1$ vector of symbols intended for users other than $A$ (the subscript $/ A$ means all other users except user $A$). $\mathbf{H}_A^\ast$ is a $T \times D$ matrix with rank $D$, and $\mathbf{I}_{ / A }^\ast$ is a $T \times (K-1)D $ matrix. Notice that,  $\mathbf{I}_{/ A}^\ast$ contains $\frac{D}{K}$ zero row vectors, since among the $T$ symbol times, $\frac{D}{K}$ slots are used for transmitting $D$ symbols only intended for user $A$.

Without quantization error, the rank of $\mathbf{I}_{/ A}^\ast$ is $(T-D)$ and thus the interference lies in a $(T-D)$-dimensional subspace. Therefore, $D$ degrees of freedom for each user can be achieved by zero forcing the interference. With quantization error, the interference will spill out of the $(T-D)$-dimensional subspace and the zero-forcing scheme cannot eliminate all the interference.
Similar to the two and three user case, we have the following lemma, which relates the  multiplexing gain with the singular values of $\mathbf{I}_{ / A }^\ast$: $\{ \sigma_1, \ldots, \sigma_{T} \}$.
\begin{lemma} \label{LemmaDOFKUser}
For the $K$ user case, the  multiplexing gain under the MAT scheme  is
\begin{align}
\text{DoF}_{\text{MAT}}(K) = \frac{D}{T} -\frac{K}{T}  \sum_{i= T-D+1}^{T-D/K } \left( \lim_{P \to \infty}  \frac{  \mathbb{E} \left[\log_2 (\sigma_i^2)\right] }{  \log_2 P } \vee (-1 ) \right). \label{EquationDoF}
\end{align}
\end{lemma}
\begin{proof}
See Appendix \ref{ProofLemmaDOFKUser}.
\end{proof}

The following theorem generalizes Theorem \ref{TheoremDOFTwoUser} and Theorem \ref{TheoremDOFThreeUser}.
\begin{theorem} \label{TheoremDOFKUser}
For the $K$ user case, if the quantization rate $Q$ is scaled as $Q= \alpha (K-1) \log_2 P$ for $\alpha >0$, the multiplexing gain under the MAT scheme is $\left( \frac{1+(K-1)\alpha}{K} \land  1 \right) \text{DoF}^\star (K) $.
\end{theorem}
\begin{proof}
See Appendix \ref{ProofTheoremDOFKUser}.
\end{proof}

When $\alpha \ge 1$, the optimal multiplexing gain without feedback rate constraint $\text{DoF}^\star (K) $ can be achieved; while if $\alpha < \alpha^\star = \frac{\frac{1}{2} + \ldots + \frac{1}{K}}{(1+ \frac{1}{2} + \ldots + \frac{1}{K}) (K-1) }$, the multiplexing gain drops to less than $1$ and outdated CSIT becomes useless.

The feedback overhead is characterized by the following theorem, which generalizes Lemma \ref{TheoremFeedbackTwoUser} and Lemma \ref{TheoremFeedbackThreeUser}.
\begin{lemma} \label{TheoremFeedbackKUser}
For the $K$ user case, if the quantization rate $Q$ is scaled as $Q=\alpha (K-1) \log_2 P$ for $\alpha>0$, then the feedback overhead with the MAT scheme is $\frac{K(K-1)(\frac{1}{2} +\ldots + \frac{1}{K}) }{(1+ \frac{1}{2} + \ldots + \frac{1}{K} )N }\alpha $.
\end{lemma}
\begin{proof}
From the transmission scheme, it can be induced that the key role played by the feedback is to ensure the interference for a particular user in future transmissions always lies in a $(T-D)$-dimensional subspace. Therefore, each receiver has to feedback its channel vector $(T-D)$ times to inform the transmitter its specific $(T-D)$-dimensional subspace. Thus, over total $NT$ symbol times, $K(T-D)$ channel vectors need feedback. It follows that the feedback rate is $\frac{K(K-1)(\frac{1}{2} +\ldots + \frac{1}{K}) }{(1+ \frac{1}{2} + \ldots + \frac{1}{K} )N }\alpha \log_2 P $ bits/slot. Finally, the theorem follows by invoking the definition of feedback overhead.
\end{proof}

Combining Theorem \ref{TheoremDOFKUser} and Theorem \ref{TheoremFeedbackKUser} gives the net DoF as
 \begin{eqnarray}
\widehat{\text {DoF} }_{\text{MAT}} (K) &=& \text{DoF} (K) - \text{FB}(K) \nonumber \\
        &=& \frac{N+\alpha (K-1) ( N - K (\frac{1}{2} + \ldots + \frac{1}{K}) ) }{ (1+ \frac{1}{2} + \ldots + \frac{1}{K}) N } \land \frac{K(N- \alpha (K-1) (\frac{1}{2} + \ldots + \frac{1}{K} ) )}{(1+ \frac{1}{2} + \ldots + \frac{1}{K})N} , \label{EquationNetDoF}
\end{eqnarray}
 which has a maximum of $\frac{K(N-(K-1)( \frac{1}{2} + \ldots + \frac{1}{K})) } { (1+ \frac{1}{2} + \ldots + \frac{1}{K})N }$.

The maximum net DoF with SISO is still $1$ and according to (\ref{EqNetDoFZF}), the maximum net DoF with ZF for $K=3$ is
\begin{align}
\widehat{\text {DoF} }_{\text{ZF}} (K)=K \left( 1- \frac{  N_{\rm fd} + K-1  }{N} \right).
\end{align}

As in the $K=2$ and $K=3$ case, we can identify a range where the MAT scheme is worthwhile. When $N \le \frac{K(K-1)( \frac{1}{2} + \ldots + \frac{1}{K} )}{ K- ( 1+ \frac{1}{2} + \ldots + \frac{1}{K}) } = (1+o(1)) K \log K $, single user SISO is better, while for $N \ge  \frac{ ( 1+ \frac{1}{2} + \ldots + \frac{1}{K}) N_{\rm fd}+ K-1 } {\frac{1}{2} + \ldots + \frac{1}{K} } = (1+o(1)) (N_{\rm fd}+ K / \log K)(1-\log^{-1}K )^{-1}$, zero forcing is preferable.  The MAT scheme is the best for all other values of $N$ when using the optimal value of $\alpha=1$.

Fig.~\ref{FigDoF} provides a visual summary of the main results. We fix the number of users $K$ and vary the block duration $N$, and plot the net DoF attained by the MAT, single user transmission (SISO), and zero-forcing (ZF). The three different channel coherence time regimes can be immediately observed: (i) short coherence time ($N \le (1+o(1)) K \log K $), (ii) moderate coherence time ($(1+o(1)) K \log K < N < (1+o(1))  (N_{\rm fd}+  K / \log K)(1-\log^{-1}K)^{-1} $), and (iii) long coherence time ($N \ge (1+o(1)) (N_{\rm fd}+ K / \log K)(1-\log^{-1}K)^{-1}$).  One would prefer to choose SISO for (i), MAT for (ii) and ZF for (iii).

\subsection{Analog vs. Digital Feedback}
An alternative way to feed back CSI is  analog feedback, where each receiver feeds back its channel vector $\mathbf{h}_r$ by explicitly transmitting $M$ complex coefficients over a unfaded additive Gaussian noise feedback channel:
\begin{align}
\mathbf{G} [t] =\sqrt{P} \mathbf{H}[t-N_{\rm fd}] + \mathbf{Z} [t],
\end{align}
where $\mathbf{G} [t]$ is the received channel state feedback and $\mathbf{Z} [t]$ is the Gaussian noise in feedback channel. As $\mathbf{H}[t]$ is composed of i.i.d. complex Gaussian with unit variance, the optimal estimator of CSI is the MMSE estimator given by
\begin{align}
\mathbf{\hat{H}} [t-N_{\rm fd}]  = \frac{\sqrt{P}}{1+ \beta P } \mathbf{G} [t],
\end{align}
where $\mathbf{\hat{H}} [t] $ is the estimator of true channel state $\mathbf{H}[t]$. Since we are interested in the scaling rate of the feedback rate with respect to $\log_2 P$ as $ P \to \infty$, the estimator noise can be neglected and the CSIT is accurate and only subject to feedback delay. Therefore, the DoF with noiseless analog feedback is the same as that with accurate CSIT, and the sum feedback overhead is $\frac{K^2}{N}$. Then, the net DoF is given by
\begin{align}
\widehat{\text {DoF} }_{\text{MAT}} (K) &= \frac{ K (N-K( 1+\frac{1}{2} + \ldots + \frac{1}{K} ) ) } { (1+ \frac{1}{2} + \ldots + \frac{1}{K})N }, \\
\widehat{\text {DoF} }_{\text{ZF}} (K) & =K \left( 1- \frac{  N_{\rm fd} + K  }{N} \right).
\end{align}
Compare it to the result with digital (quantized) feedback, we see that the net DoF with analog feedback is almost the same as the maximum net DoF ($\alpha=1$) with digital feedback. Therefore, the tradeoff between coherence time, feedback rate/delay, and the transmission techniques shown in Fig.~\ref{FigDoF} remains the same with analog feedback. Note that the digital feedback appears to be more flexible than analog feedback, since it can adjust $\alpha$ to meet the feedback rate constraints and achieve a gradual degradation of net DoF while analog feedback is only feasible when feedback overhead $\frac{K^2}{N}$ is supportable in the feedback channel.

\subsection{Design Guidelines}
\label{sec:LTE}
This subsection translates the previous analytical results into rough design guidelines for a real-world system.  To be concrete, we adopt the parameters used in the 3GPP LTE standard \cite{LTEBook}. The carrier frequency is chosen to be $f_c=2.1$ GHz and resources are allocated to users in ``resource blocks'' in the time-frequency grid consisting of $12$ subcarriers, spanning $180$ KHz in frequency, over $14$ OFDM symbols, which spans $1$ msec in time. Therefore, a data symbol slot is effectively $ T_s= 1/168$ msec since there are $12 \times 14 = 168$ symbols sent in a msec. The typical CSI feedback delay is assumed to be an LTE frame, which is $10$ msec, so $N_{\rm fd} = 10 \times 168 =1680$ symbols.  Assuming the standard relation between channel coherence time, Doppler spread, and user velocity $v$, we have $v = \frac{c}{f_c N T_s}$m/s, where $c$ is the speed of light. Then, based on the results in the previous section, approximate regimes where the MAT scheme achieves a net DOF gain are summarized in Table~\ref{tapcap}. We caution against taking this table too literally, since other factors not modeled in this paper may play a role in the regimes of optimality, but nevertheless it appears that the MAT scheme provides a net DoF gain for a very large range of mobility.  The upper limit is where the mobility is so high that it is better to switch to simple single user transmission, whereas the lower bound on velocity represents the crossing point over to ZF precoding.  We observe that as the number of antennas (and thus users) increases, ZF precoding takes on increasing role since its maximum achievable DoF is $K$ vs. MAT's $K/\log K$ and the $\log K$ gap becomes more significant.

\appendix
{
\subsection{Proof of Lemma \ref{LemmaDOFTwoUser}} \label{ProofLemmaDOFTwoUser}
\begin{proof}
As can be seen in (\ref{EquationChannelTwoUser}), in order to achieve the maximal multiplexing gain, receiver A must attempt to zero-force the interference, i.e.,
\begin{align}
\mathbf{U}_A^\ast \mathbf{y}_A =
\mathbf{U}_A^\ast \mathbf{H}^{\ast}_A
\begin{bmatrix}
u_A \\
v_A \\
\end{bmatrix}
+ \mathbf{U}_A^\ast \mathbf{I}^{\ast}_B
\begin{bmatrix}
u_B \\
v_B \\
\end{bmatrix},
\end{align}
where $\mathbf{U}_A^\ast$ is a $ 2 \times 3$ zero-forcing matrix. The rank of $\mathbf{U}_A^\ast$ must be $2$ to recover $u_A$ and $v_A$. Then, the average throughput of receiver $A$ can be derived as
\begin{align}
R_A(P) = \frac{1}{3} \mathbb{E} \left[ \log_2 \frac{\det \left( \mathbf{I}_2 + \frac{P}{2} \mathbf{H}_A \mathbf{U}_A \mathbf{U}_A^\ast \mathbf{H}_A^\ast  + \frac{P}{2} \mathbf{I}_B \mathbf{U}_A \mathbf{U}_A^\ast \mathbf{I}_B^\ast  \right) }{ \det \left( \mathbf{I}_2 + \frac{P}{2} \mathbf{I}_B \mathbf{U}_A \mathbf{U}_A^\ast \mathbf{I}_B^\ast  \right)} \right].
\end{align}
Define the singular value decomposition $\mathbf{I}_B= \mathbf{U} \mathbf{\Sigma} \mathbf{V}^\ast$, where $\mathbf{\Sigma} = \text{diag} \{ \sigma_1, \sigma_2 \}$ and $\mathbf{U}_A^\ast$ is chosen by canceling as much interference as possible. Using matrix analysis, we can derive
\begin{align}
& \det \left( \mathbf{I}_2 + \frac{P}{2} \mathbf{I}_B \mathbf{U}_A \mathbf{U}_A^\ast \mathbf{I}_B^\ast  \right) \nonumber \\
& \overset{(a)} {=} \det \left( \mathbf{I}_2 + \frac{P}{2}  \mathbf{U}_A^\ast \mathbf{I}_B^\ast \mathbf{I}_B \mathbf{U}_A  \right) \nonumber \\
& \overset{(b)} {=} (1+ \frac{P}{2} \lambda^2_1)(1+ \frac{P}{2}\lambda^2_2) \nonumber \\
& \overset {(c)} {\ge}  1+  \frac{P}{2} ( \lambda^2_1 + \lambda^2_2) \nonumber \\
& \overset{(d)} {=}  1 +  \frac{P}{2} \| \mathbf{I}_{B} \mathbf{U}_A \|_F \nonumber \\
& \overset{(e)}{\ge} 1 +  \frac{P}{2} \sigma^2_2,
\end{align}
where $(a)$: follows from the fact that $\det (\mathbf{I} + \mathbf{A}\mathbf{B}) = \det(\mathbf{I} + \mathbf{B}\mathbf{A} )$.

$(b)$: follows from the definition that $\lambda_1,\lambda_2$ are the singular values of $\mathbf{I}_{B} \mathbf{U}_A$.

$(c)$: follows from neglecting the remaining nonnegative parts $\lambda^2_1\lambda^2_2$.

$(d)$: follows from the definition of Frobenius norm $\| \cdot \|_F$.

$(e)$: follows from the fact that $\min_{\mathbf{U}_A: \mathbf{U}_A^\ast \mathbf{U}_A = \mathbf{I}} \| \mathbf{I}_{B} \mathbf{U}_A \|_F = \sigma_2^2$.

Choose $\mathbf{U}_A$ to be the last $2$ columns of $\mathbf{V}$, i.e., $\mathbf{U}_A= \mathbf{V}(2:3)$, the lower bound is achieved. The interference power therefore becomes
\begin{align}
\det \left( \mathbf{I}_2 + \frac{P}{2} \mathbf{I}_B \mathbf{U}_A \mathbf{U}_A^\ast \mathbf{I}_B^\ast  \right) = 1+ \frac{P}{2} \sigma_{2}^2. \label{EqInterPowerTwoUser}
\end{align}
Then, by the definition of multiplexing gain in (\ref{EquationDefDoF}), we have $\text{DoF}_A (2)$
\begin{align}
 & {=}  \frac{1}{3} \lim_{P \to \infty} \left(  \frac{\mathbb{E} \left[ \log_2 \det \left( \mathbf{I}_2 + \frac{P}{2} \mathbf{H}_A \mathbf{U}_A \mathbf{U}_A^\ast \mathbf{H}_A^\ast +\frac{P}{2} \mathbf{I}_B \mathbf{U}_A \mathbf{U}_A^\ast \mathbf{I}_B^\ast  \right) \right] }{ \log_2 P}- \frac{\mathbb{E} \left[ \log_2 \det \left( \mathbf{I}_2 + \frac{P}{2} \mathbf{I}_B \mathbf{U}_A \mathbf{U}_A^\ast \mathbf{I}_B^\ast  \right) \right]}{\log_2 P } \right) \nonumber \\
 &\overset{(a)} {=}\frac{2}{3} - \frac{1}{3} \lim_{P \to \infty} \frac{\mathbb{E} \log_2 \det \left( \mathbf{I}_2 + \frac{P}{2} \mathbf{I}_B \mathbf{U}_A \mathbf{U}_A^\ast \mathbf{I}_B^\ast  \right)}{\log_2 P } \nonumber \\
 & \overset{(b)}{=} \frac{2}{3}- \frac{1}{3} \left( \lim_{P \to \infty} \frac{ \mathbb{E} \left[\log_2 ( P \sigma_2^2)\right] }{  \log_2 P } \vee 0 \right)\nonumber \\
& {=} \frac{1}{3} - \frac{1}{3} \left( \lim_{P \to \infty} \frac{ \mathbb{E} \left[\log_2 (  \sigma_2^2)\right] }{  \log_2 P } \vee (-1) \right),
\end{align}
where $(a)$ follows from the fact that the rank of $ \mathbf{H}_A \mathbf{U}_A$ is $2$ almost surely and that the signal power dominates the interference power when $P \to \infty$; $(b)$ follows from (\ref{EqInterPowerTwoUser}) and let $P \to \infty$. Finally, the proof is completed by considering receiver B similarly.
\end{proof}

\subsection{Proof of Theorem \ref{TheoremDOFTwoUser}} \label{ProofTheoremDOFTwoUser}
\begin{proof}
From Lemma \ref{LemmaDOFTwoUser}, the multiplexing gain can be derived as
\begin{align}
\text{DoF} (2)
& =   \frac{2}{3}- \frac{2}{3} \left( \lim_{P \to \infty}  \mathbb{E} \left[ \frac{ \log_2 ( \sigma_2^2) }{  \log_2 P } \right] \vee (-1) \right) \nonumber \\
& \overset{(a)} {=}   \frac{2}{3}- \frac{2}{3} \left( \lim_{P \to \infty}   \frac{ \log_2 (  \mathbb{E} \left[\sigma_2^2\right] ) }{  \log_2 P } \vee (-1) \right) \nonumber \\
 & \overset{(b)}{=} \frac{2}{3} - \frac{2}{3} \left( \lim_{P \to \infty} \frac{ \log_2 ( \mathbb{E} \left[ \sin^2 \theta \right] )  } { \log_2 P }  \vee (-1)  \right) \nonumber \\
 &  \overset{(c)} {=} \frac{2}{3} + \frac{2}{3} \left( \lim_{P \to \infty} \frac{B}{\log_2 P} \land 1 \right)\nonumber \\
 & = \frac{2}{3} + \frac{2}{3} \left( \alpha \land 1 \right) ,
\end{align}
where $(a)$ follows from Jensen's inequality and $\log(x)$ being a concave function, which gives the inequality $\ge$. Moreover, since $P \to \infty$ and thus $B \to \infty$, $\sigma_2$ converges to $0$ almost surely. Therefore, the inequality $\ge$ becomes an equality.

$(b)$ follows from the fact that
\begin{align}
\sigma_{2}^2= 1- \cos \theta = \Theta \left( \sin^2 \theta \right), \label{EqSigmaTwoUser}
\end{align}
which is proved in Appendix \ref{ProofEqSigmaTwoUser}.

$(c)$ follows from the lower and upper bound of $ \mathbb{E} \left[ \sin^2 \theta \right]$ in (\ref{LemmaQuantization}).
\end{proof}

\subsection{Proof of Lemma \ref{LemmaDOFThreeUser}} \label{ProofLemmaDOFThreeUser}
\begin{proof}
Focus on receiver A first. Similar to the two user case, the average throughput of receiver $A$ is given by
\begin{align}
R_A(P) = \frac{1}{11} \mathbb{E} \left[\log_2 \frac{\det \left( \mathbf{I}_6 + \frac{P}{3} \mathbf{H}_A \mathbf{U}_A \mathbf{U}_A^\ast \mathbf{H}_A^\ast  +  \frac{P}{3} \mathbf{I}_{B,C} \mathbf{U}_A \mathbf{U}_A^\ast \mathbf{I}_{B,C}^\ast \right) }{ \det \left( \mathbf{I}_{12} + \frac{P}{3} \mathbf{I}_{B,C} \mathbf{U}_A \mathbf{U}_A^\ast \mathbf{I}_{B,C}^\ast  \right)} \right].
\end{align}
Choose $\mathbf{U}_A$ as the last $6$ columns of $\mathbf{V}$: $\mathbf{U}_A= \mathbf{V}(6:11)$, where $\mathbf{I}_{BC}= \mathbf{U} \mathbf{\Sigma} \mathbf{V}^\ast$ is the singular value decomposition. Then, the interference power can be simplified as
\begin{align}
\det \left( \mathbf{I}_{12} + \frac{P}{3} \mathbf{I}_{B,C} \mathbf{U}_A \mathbf{U}_A^\ast \mathbf{I}_{B,C}^\ast  \right) = \prod_{i=6}^{11} (1+ \frac{P}{3} \sigma_i^2)= \prod_{i=6}^{9} (1+ \frac{P}{3} \sigma_i^2),
\end{align}
where the last equality follows from the fact that $\sigma_{10}=\sigma_{11}=0$ since $\mathbf{I}_{B,C}$ contains two zero row vectors. Then the lemma naturally follows by considering the definition of the multiplexing gain and the fact that the rank of $\mathbf{H}_A \mathbf{U}_A$ is $6$.
\end{proof}

\subsection{Proof of Theorem \ref{TheoremDOFThreeUser}} \label{ProofTheoremDOFThreeUser}
\begin{proof}
From Lemma \ref{LemmaDOFThreeUser}, we have
\begin{align}
\text{DoF}(3) & = \frac{6}{11} - \frac{3}{11}  \sum_{i=6}^9 \left( \lim_{P \to \infty}  \frac{  \mathbb{E} \left[\log_2 (\sigma_i^2)\right] }{  \log_2 P } \vee (-1) \right)\nonumber \\
& =\frac{6}{11} - \frac{3}{11} \sum_{i=6}^9  \left( \lim_{P \to \infty}  \frac{  \log_2 ( \mathbb{E} \left[ \sigma_i^2 \right])  }{  \log_2 P } \vee (-1) \right)\nonumber \\
& \overset{(a)} {=}  \frac{6}{11}- \frac{3}{11} \sum_{i=6}^9 \left( \lim_{P \to \infty}  \frac{  \log_2 ( \mathbb{E} \left[ \sin^2 \theta)\right] }{  \log_2 P } \vee (-1) \right) \nonumber \\
& = \frac{6}{11}+ \frac{12}{11} \left( \frac{B}{ 2 \log_2 P} \land 1 \right)\nonumber \\
& = \frac{6}{11}+ \frac{12}{11} \left( \alpha \land 1 \right),
\end{align}
where $(a)$ follows from the fact that
\begin{align}
\mathbb{E} [\sigma_i^2] =  \Theta \left(  \mathbb{E} [\sin^2 \theta] \right), \text{ for } i=6,\ldots,9, \label{EqSigmaThreeUser}
\end{align}
which is proved in Appendix \ref{ProofEqSigmaKUser}.
\end{proof}

\subsection{Proof of Lemma \ref{LemmaDOFKUser} } \label{ProofLemmaDOFKUser}
\begin{proof}
Focus on receiver A first. As can be seen in (\ref{EquationChannelKUser}), in order to achieve the maximal multiplexing gain, receiver A must try to zero-forcing the interference, i.e.,
\begin{align}
\mathbf{U}^\ast_A \mathbf{y}_A = \mathbf{U}^\ast_A\mathbf{H}_A^\ast \mathbf{u}_A
+  \mathbf{U}^\ast_A \mathbf{I}_{/ A}^\ast  \mathbf{u}_{/A},
\end{align}
where $\mathbf{U}_A^\ast$ is a $ D \times T$ zero-forcing matrix of rank $D$. Then, the average throughput of user $A$ is
\begin{align}
R_A(P) = \frac{1}{T} \mathbb{E} \left[ \log_2 \frac{\det \left( \mathbf{I}_{D} + \frac{P}{K} \mathbf{H}_A \mathbf{U}_A \mathbf{U}_A^\ast \mathbf{H}_A^\ast +  \frac{P}{K} \mathbf{I}_{/ A} \mathbf{U}_A \mathbf{U}_A^\ast \mathbf{I}_{ / A}^\ast \right) }{ \det \left( \mathbf{I}_{(K-1)D} + \frac{P}{K} \mathbf{I}_{/ A} \mathbf{U}_A \mathbf{U}_A^\ast \mathbf{I}_{ / A}^\ast  \right)} \right].
\end{align}
Choose $\mathbf{U}_A$ as the last $D$ columns of $\mathbf{V}$: $\mathbf{U}_A= \mathbf{V}(T-D+1:T)$, where $\mathbf{I}_{/ A}= \mathbf{U} \mathbf{\Sigma} \mathbf{V}^\ast$ is the singular value decomposition. Then, the interference power becomes
\begin{align}
\det \left( \mathbf{I}_{(K-1)D} + \frac{P}{K} \mathbf{I}_{/ A} \mathbf{U}_A \mathbf{U}_A^\ast \mathbf{I}_{/ A}^\ast  \right) = \prod_{i=T-D+1}^{T } (1+ \frac{P}{K} \sigma_i^2)= \prod_{i=T-D+1}^{T-D / K } (1+ \frac{P}{K} \sigma_i^2), \label{EqInterPowerKUser}
\end{align}
where the last equality follows from the fact that $\mathbf{I}_{/ A}$ contains $\frac{D}{K}$ zero row vectors.

Then, by the definition of multiplexing gain in (\ref{EquationDefDoF}), we have
\begin{align}
\text{DoF}_A (K) &= \lim_{P \to \infty} \frac{R_A(P)}{ \log_2 P} \nonumber \\
& \overset{(a)}{=} \frac{D}{T} - \frac{1}{T} \lim_{P \to \infty} \mathbb{E} \left[ \log_2 \det \left( \mathbf{I}_{(K-1)D} + \frac{P}{K} \mathbf{I}_{/ A} \mathbf{U}_A \mathbf{U}_A^\ast \mathbf{I}_{ / A}^\ast  \right) \right] \nonumber \\
& \overset{(b)}{=} \frac{D}{T}- \frac{1}{T} \sum_{i= T-D+1}^{T-D/K } \left( \lim_{P \to \infty}   \frac{ \mathbb{E} \left[\log_2 ( P \sigma_i^2)\right] }{  \log_2 P } \vee 0 \right) \nonumber \\
&=\frac{D}{KT} - \frac{1}{T} \sum_{i= T-D+1}^{T-D/K } \lim_{P \to \infty}  \left(  \frac{  \mathbb{E} \left[\log_2 (\sigma_i^2)\right] }{  \log_2 P } \vee (-1) \right),
\end{align}
where $(a)$ follows from the fact that $\mathbf{H}_A \mathbf{U}_A$ is of rank $D$ and that the signal power dominates the interference power when $P \to \infty$; $(b)$ follows from (\ref{EqInterPowerKUser}) and let $P \to \infty$. Finally, the proof can be readily completed by considering all the users.
\end{proof}

\subsection{Proof of Theorem \ref{TheoremDOFKUser}} \label{ProofTheoremDOFKUser}
\begin{proof}
From Lemma \ref{LemmaDOFKUser}, we have
\begin{align}
\text{DoF}(K) & = \frac{D}{T} - \frac{K}{T}  \sum_{i= T-D+1}^{T-D/K } \left( \lim_{P \to \infty}    \frac{\mathbb{E} \left[\log_2 (\sigma_i^2)\right] }{  \log_2 P } \vee (-1) \right) \nonumber \\
& \overset{(a)}{=} \frac{D}{T} - \frac{K}{T}   \sum_{i= T-D+1}^{T-D/K } \left(\lim_{P \to \infty}   \frac{ \log_2 ( \mathbb{E} \left[\sigma_i^2 \right]) }{  \log_2 P } \vee (-1) \right) \nonumber \\
& \overset{(b)} {=} \frac{D}{T} - \frac{(K-1)D }{T}  \left( \lim_{P \to \infty} \frac{ \log_2 (\mathbb{E} \left[ \sin^2 \theta \right] )  } { \log_2 P } \vee (-1) \right) \nonumber \\
& \overset{(c)} {=} \frac{D}{T} + \frac{(K-1)D }{T} \left( \lim_{P \to \infty} \frac{B}{(K-1)\log_2 P } \land 1 \right) \nonumber \\
& = \frac{D}{T} + \frac{(K-1)D }{T} \left( \alpha \land 1 \right) ,
\end{align}
where $(a)$ follows from the Jensen's inequality and the fact that $\log(x)$ is a concave function: we have the inequality $\ge$. Moreover, since $P \to \infty$, $B \to \infty$, $\sigma_i$ converges to zero almost surely. Therefore, the inequality $\ge$ becomes equality.

$(b)$ follows from the fact that
\begin{align}
\mathbb{E}[\sigma_i^2] =  \Theta \left( \mathbb{E}[\sin^2 \theta] \right), \text{ for } i=T-D+1, \ldots, T- \frac{D}{K}, \label{EqSigmaKUser}
\end{align}
which is proved in Appendix \ref{ProofEqSigmaKUser}.

$(c)$ follows from the lower and upper bound of $ \mathbb{E} \left[ \sin^2 \theta \right]$  in (\ref{LemmaQuantization}).
\end{proof}

\subsection{Proof of (\ref{EqSigmaTwoUser})} \label{ProofEqSigmaTwoUser}
\begin{proof}
To find the $\sigma_2$, it suffices to consider the last two rows of $\mathbf{I}_B^\ast$, denoted by $\hat{ \mathbf{I}}_B^\ast$. Let us define a new $2$-dimensional unit norm vector $\mathbf{\alpha}$ such that its angles to the two rows of $ \mathbf{I}_B^\ast $ are the same. Also, define $\mathbf{\alpha}^\bot $ as the $2$-dimensional unit norm vector which is orthogonal to $\mathbf{\alpha}$.  Then we have the following singular value decomposition of $\hat{ \mathbf{I}}_B^\ast$:
\begin{align} \label{EqSVDTwoUser}
\hat{ \mathbf{I}}_B^\ast = \begin{bmatrix}
\frac{1}{\sqrt{2}} & \frac{1}{\sqrt{2}} \\
\frac{1}{\sqrt{2}} & -\frac{1}{\sqrt{2}}
\end{bmatrix}
\begin{bmatrix}
\sqrt{2} \cos \frac{\theta}{2}  & 0 \\
0 & \sqrt{2} \sin \frac{\theta}{2}
\end{bmatrix}
\begin{bmatrix}
\mathbf{\alpha}^\ast \\
\left(\mathbf{\alpha}^\bot \right)^\ast
\end{bmatrix}.
\end{align}
Therefore, $\sigma^2_2 = 2 \sin^2 \frac{\theta}{2}=\Theta \left( \sin^2 \theta \right)$.
\end{proof}

\subsection{Proof of (\ref{EqSigmaThreeUser}) and (\ref{EqSigmaKUser})} \label{ProofEqSigmaKUser}
\begin{proof}
It suffices to consider the nonzero rows of interference matrix $\mathbf{I}_{/ A }^\ast$. Thus, let us remove the $\frac{D}{K}$ zero rows and still denote it as $\mathbf{I}_{/ A }^\ast$ for ease of notation. First prove the upper bound of singular values of $\mathbf{I}_{/A}^\ast$. Denote $\mathbf{I}_{/ A }^\ast$ with no quantization error as $\bar{\mathbf{I}}_{/ A }^\ast$ and define $\mathbf{E}=\bar{\mathbf{I}}_{/ A }^\ast-\mathbf{I}_{/ A }^\ast $. Let $\mathbf{E}_i$ denote rows of $\mathbf{E}$ for $ i=1, \ldots, T-\frac{D}{K}$.

From the perturbation bounds for the singular values of a matrix due to Weyl \cite{Weyl}, we have
\begin{eqnarray}
|\sigma_i- \bar{\sigma}_i | \le \| \mathbf{E} \|_F, \text{ for } i=T-D+1, \ldots, T- \frac{D}{K}
\end{eqnarray}
where $\bar{\sigma}_i$ is the singular value of $\bar{\mathbf{I}}_{/ A }^\ast$. Since $\bar{\mathbf{I}}_{/ A }^\ast$ is of rank $T- \frac{D}{K}$, we have that $\bar{\sigma}_i =0$ for $i=T-D+1, \ldots, T- \frac{D}{K}$.

Moreover, $\mathbb{E} \left[ \| \mathbf{E}_i \|^2_2 \right] \le 4 \mathbb{E} \left[ \sin^2 \frac{\theta}{2} \right]$, for $i=1, \ldots, T- \frac{D}{K}$. Therefore, $\mathbb{E} \|E\|^2_F \le 4 ( T- \frac{D}{K}) \mathbb{E} \left[ \sin^2 \frac{\theta}{2} \right]$. It follows then
\begin{eqnarray}
\mathbb{E} \left[ \sigma^2_i \right] & \le& 4(T- \frac{D}{K}) \mathbb{E} \left[ \sin^2 \frac{\theta}{2} \right] \nonumber \\
&=& O\left( \mathbb{E} \left[\sin^2 \theta \right] \right), \text{ for } i=1, \ldots, T- \frac{D}{K}.
\end{eqnarray}


Next prove the lower bound, i.e., $ \mathbb{E} \left[ \sigma^2_{T- \frac{D}{K}} \right]= \Omega \left( \mathbb{E} \left[\sin^2 \theta \right] \right)$. Let $\mathbf{I}_i$ denote rows of $\mathbf{I}_{/A}^\ast$, for $i=1,\ldots, T- \frac{D}{K}$ and $\mathbf{S}_i$ denote the space spanned by all the rows $\mathbf{I}_1, \ldots, \mathbf{I}_{T- \frac{D}{K}}$ other than $\mathbf{I}_i$. Define
\begin{align}
\text{dist} (\mathbf{I}_i, \mathbf{S}_i) := \min_{\mathbf{X}_i \in \mathbf{S}_i}   \| \mathbf{I}_i-\mathbf{X}_i\|_2
\end{align}
Due to the quantization error, $ \text{dist} (\mathbf{I}_i, \mathbf{S}_i)  \ge 2 \sin \frac{\theta}{2}$ almost surely when $\theta$ is sufficiently small. Furthermore, from the {\it negative second moment identity} in \cite{Tao08} we have
\begin{align}
\sum_{i=1}^{T- \frac{D}{K}} \sigma^{-2}_i = \sum_{i=1}^{T- \frac{D}{K}} \text{dist} (\mathbf{I}_i, \mathbf{S}_i)^{-2}.
\end{align}
Then it follows that
\begin{align}
\sigma_{T- \frac{D}{K}}^{-2} & \le  \sum_{i=1}^{T- \frac{D}{K}} \sigma^{-2}_i
  =  \sum_{i=1}^{T- \frac{D}{K}} \text{dist} (\mathbf{I}_i, \mathbf{S}_i)^{-2} \nonumber \\
 & \le  \sum_{i=1}^{T- \frac{D}{K}} \frac{1}{4 \sin^2 \theta}
  =  \frac{T- \frac{D}{K}}{4 \sin^2 \theta}.
\end{align}
Therefore, we have
\begin{align}
\sigma^2_{T- \frac{D}{K}} &\ge \frac{4 \sin^2 \theta}{T- \frac{D}{K}}
= \Omega \left( \sin^2 \theta \right).
\end{align}
We conclude that $ \mathbb{E}[\sigma_i^2] =  \Theta \left( \mathbb{E}[\sin^2 \theta] \right), \text{ for } i=T-D+1, \ldots, T- \frac{D}{K}$.
\end{proof}
}

\bibliography{OutdatedCSIT}

\newpage

\begin{table}
\centering
\caption{Approx. range of optimality of the MAT scheme with LTE-like parameters, see Sec. \ref{sec:LTE}}
\label{tapcap}
\begin{tabular}{l|l|l|l}
\hline
Number of antennas $K$   &   Coherence time  $N$    &  Coherence time $T_c$ (msec)  &  Velocity $v$ (km/hr) \\
\hline
2 &    $ 2 \le N \le 5000 $  &   $ T_c \le 30 $ &  $ v \ge 17  $  \\
\hline
4 &   $ 7 \le N \le 3200 $ &  $ 0.04 \le T_c \le 20  $ & $27 \le v \le 12,000$ \\
\hline
16 &  $ 46 \le N \le 2400 $ & $0.3 \le T_c \le 14 $ &  $36 \le v \le 1900 $\\
\hline
\end{tabular}
\end{table}

\begin{figure}
\centering
\includegraphics[width=7.0in]{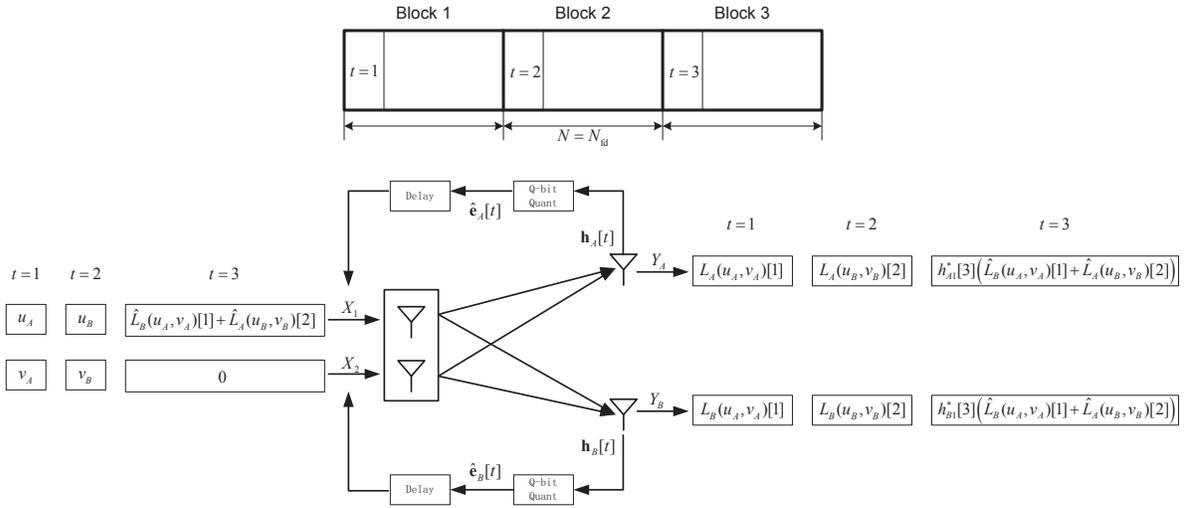}
\centering
\caption{Block diagram of the MAT scheme with quantization error for $K=2$.}
\label{FigSysModel}
\end{figure}

\newpage

\begin{figure}
\centering
\includegraphics[width=5.0in]{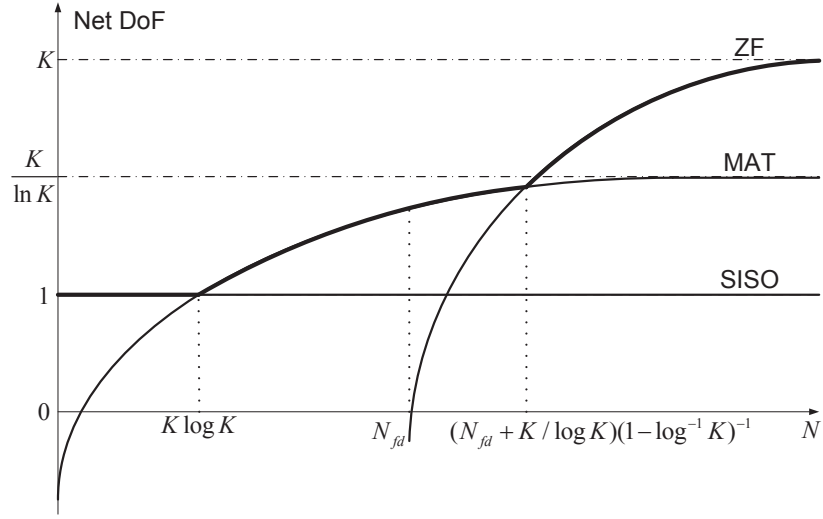}
\centering
\caption{The net DoF for MAT, SISO and multiuser ZF precoding versus the coherence block length $N$ with large $K$.  Note the is figure is not to scale (the range of optimality for MAT would be much wider).}
\label{FigDoF}
\end{figure}

\newpage

\begin{figure}
\centering
\includegraphics[width=5.0in]{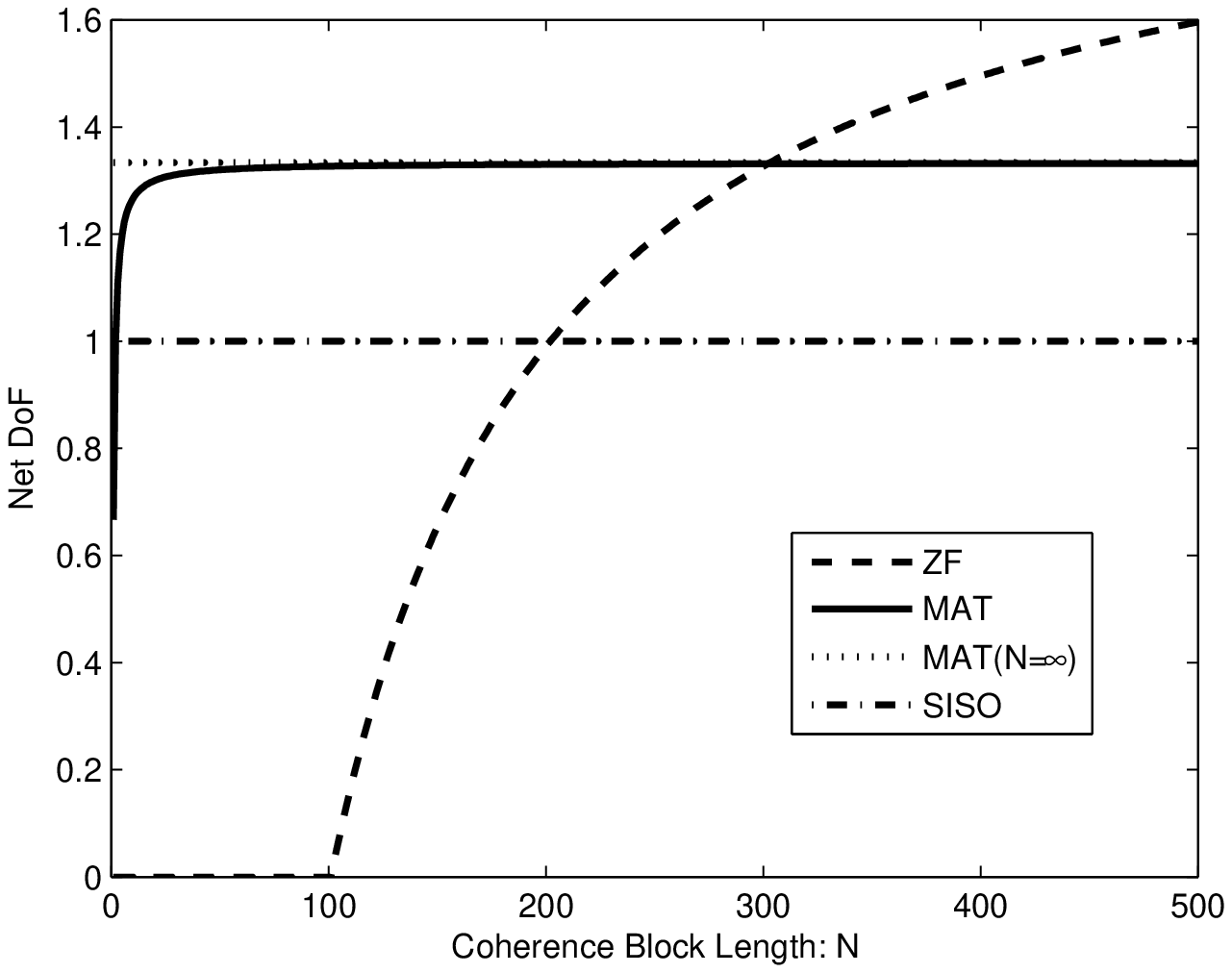}
\centering
\caption{The net DoF of MAT, SISO and ZF for $K = 2$ with varying coherence block length $N$.}
\label{FigNetDoF2}
\end{figure}

\begin{figure}
\centering
\includegraphics[width=5.0in]{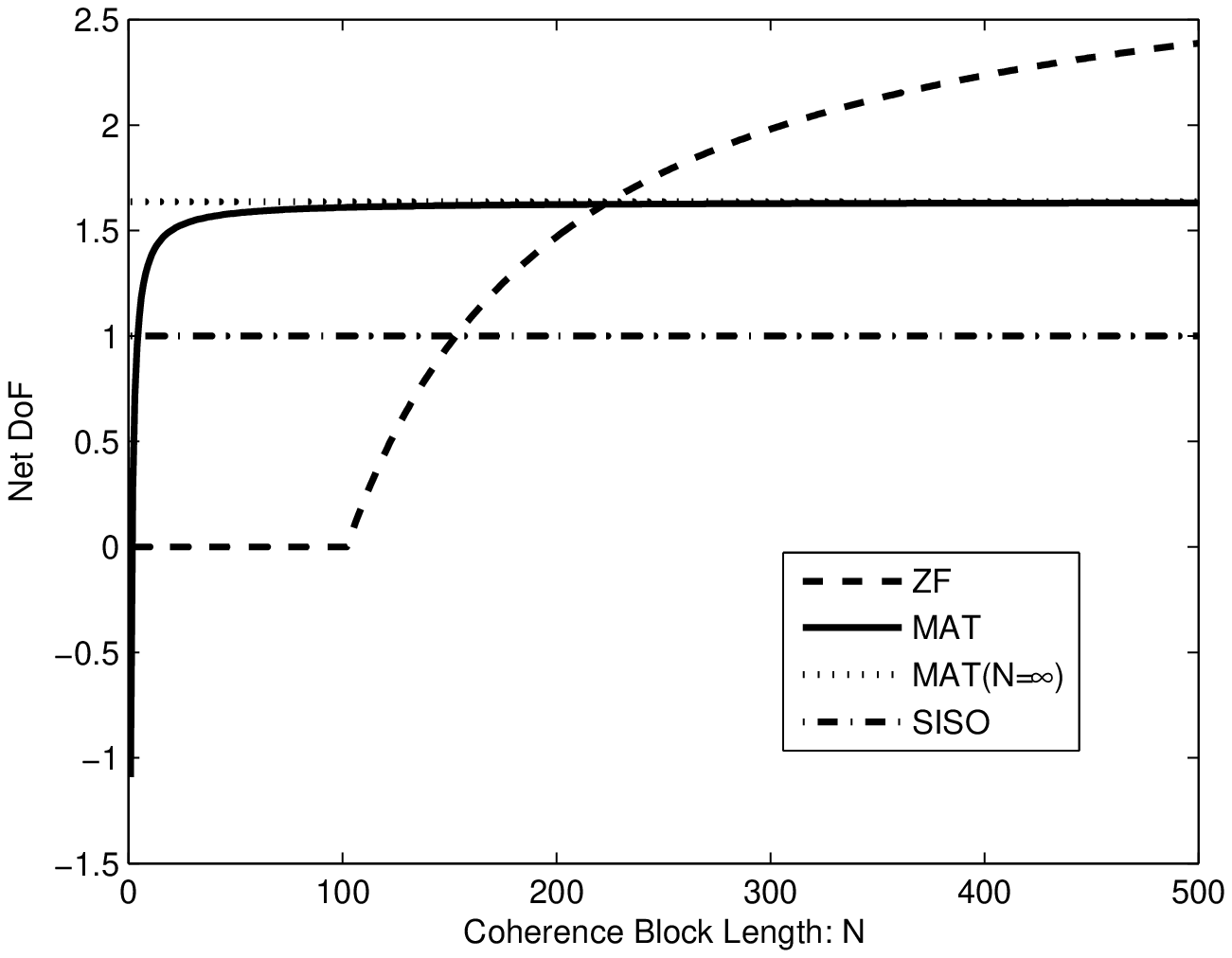}
\centering
\caption{The net DoF of the MAT, SISO and ZF for $K=3$ with varying coherence block length $N$.}
\label{FigNetDoF3}
\end{figure}

\end{document}